\documentclass{article} 
\pdfoutput=1
\usepackage{graphicx}      
\usepackage{authblk}

\usepackage[cmex10]{amsmath}
\usepackage{amssymb,graphics,amscd,amsopn,amstext,amsfonts,amsthm}

\newtheorem{proposition}{Proposition}
\newtheorem{theorem}{Theorem}
\newtheorem{corollary}{Corollary}
\newtheorem{assumption}{Assumption}
\newtheorem{lemma}{Lemma}
\newtheorem{remark}{Remark}
\newtheorem{problem}{Problem}

\newcommand{\be}{\begin{equation}}
\newcommand{\ee}{\end{equation}}
\newcommand{\bea}{\begin{eqnarray}}
\newcommand{\eea}{\end{eqnarray}}
\newcommand{\beas}{\begin{eqnarray*}}
\newcommand{\eeas}{\end{eqnarray*}}
\newcommand{\nn}{\nonumber}
\newcommand{\bbm}{\begin{bmatrix}}
\newcommand{\ebm}{\end{bmatrix}}
\newcommand{\matl}{\left[ \begin{array}}
\newcommand{\matr}{\end{array} \right]}

\newcommand{\mC}{\mathrm{C}}
\newcommand{\cC}{\mathcal{C}}
\newcommand{\cD}{\mathcal{D}}

\newcommand{\bR}{\mathbb{R}}
\newcommand{\bS}{\mathbb{S}}
\newcommand{\bW}{\mathbb{W}}

\newcommand{\cB}{\mathcal{B}}
\newcommand{\cK}{\mathcal{K}}
\newcommand{\cF}{\mathcal{F}}
\newcommand{\cG}{\mathcal{G}}

\newcommand{\cM}{\mathcal{M}}
\newcommand{\cN}{\mathcal{N}}

\newcommand{\Delt}{\Delta t}

\newcommand{\mrm}{\mathrm}

\newcommand{\T}{^{\mbox{\small T}}}
\newcommand{\msf}{\mathsf}
\newcommand{\sF}{\msf{F}}
\newcommand{\sG}{\msf{G}}

\begin{document}

\title{Data-Driven Discrete-time Control with H\"{o}lder-Continuous Real-time Learning}

\author[$\dagger$]{A. K. Sanyal} 
\affil[$\dagger$]{Department of Mechanical and Aerospace Engineering, Syracuse University, Syracuse, 
NY 13104, ph: (315) 443--0466, email: aksanyal@syr.edu} 
\maketitle

\begin{abstract}
This work provides a framework for data-driven control of discrete time systems with unknown input-output 
dynamics and outputs controllable by the inputs. This framework leads to stable and robust real-time control of 
the system such that a feasible output trajectory can be tracked. This is made possible by rapid real-time stable 
learning of the unknown dynamics using H\"{o}lder-continuous learning schemes that are designed as discrete-time 
stable disturbance observers. This observer learns from prior input-output history and it ensures finite-time 
stable convergence of model estimation errors to a bounded neighborhood of the zero vector if the system is 
known to be Lipschitz-continuous with respect to outputs, inputs, internal parameters and states, and time. In 
combination with nonlinearly stable controller designs, this makes the proposed framework nonlinearly stable 
and robust to disturbances, model uncertainties, and unknown measurement noise. Nonlinear stability and 
robustness analyses of the observer and controller designs are carried out using discrete Lyapunov analysis. 
H\"{o}lder-continuous Finite-time stable observer and controller designs in this framework help to prove 
robustness of these schemes and guaranteed convergence of outputs to a bounded neighborhood of the desired 
output trajectory. A numerical experiment on a nonlinear second-order system 
demonstrates the performance of this discrete nonlinear model-free control framework. 
\end{abstract}



\section{Introduction}\label{sec: intro}
\vspace*{-1mm}
Data-driven control approaches are used for feedback control of systems with uncertain or unknown 
input-output behavior. When only input-output behavior of the system is observable, then output 
regulation to a desired set point or output trajectory tracking has to be based on data-driven 
(model-free) controller and observer designs. This work provides a nonlinear data-driven control 
framework for output tracking of systems for which the measured output variables can be controlled 
by the applied inputs, but the model describing the input-output relation is unknown or uncertain. 
The main contribution of this work is that it provides definite (quantifiable) 
guarantees on nonlinear stability and robustness in real-time learning of the unknown input-output 
dynamics and controlling the output along feasible prescribed trajectories. The development and 
implementation of this framework is carried out in discrete time for ease of computer implementation.

A majority of linear and nonlinear control approaches are model-based, for which a model of the 
dynamics of the system being controlled is necessary. However, as the number and variety of systems  
to which control theory is applied increases, uncertainties and difficulties in modeling need to be 
overcome. Of particular interest for this work is the large class of (nonlinear) systems with uncertain or 
unmodeled dynamics that have to be controlled in real time. This class of systems includes, for example, 
autonomous vehicles, walking robots, and electronic medical implants. For such systems, data driven 
(i.e., model-free) control techniques may be used for feedback control in real-time. In the last 15 years, 
the term ``model-free control" for uncertain systems has been used in different senses and settings 
in the published literature. 
These settings are quite varied, and range from ``classic" PIDs to feedback control using techniques 
from neural nets, fuzzy logic, and soft computing to learn the uncertainties in the dynamics, e.g.,
in~\cite{kelbha08,kilkrst06,doco10,sata11,renbig17}. A model-free control framework based on 
classical control, termed the ``intelligent PID" (or ``iPID") scheme, was proposed in~\cite{fljosi08,FlJo13}. 
The iPID framework uses an {\em ultra-local model} to describe the unknown input-output dynamics, 
and estimates and uses this model for feedback control. In addition, if the system is known to be 
differentially flat for the selected outputs~\cite{fllemaro95}, then a state trajectory can also be tracked 
and uncertainties in the input-state dynamics can also be estimated over time from the measured outputs 
using model-free filtering techniques~\cite{fljosi08,trsiba07}. However, in the iPID framework, the 
ultra-local model is estimated by a linear filtering scheme assuming measurements at a sufficiently high 
sampling frequency~(\cite{Mboup2009,griz17}). 
Some applications where similar model-free control techniques have been considered are given in, 
e.g.,~\cite{machines-mfc,ydn16,viba11,chgagu11}. More recently, a data-enabled predictive control 
(``DeePC") method was formulated for data-driven control of unmodeled/uncertain systems that is 
analogous to the classical model-predictive control (MPC) technique for model-based control of linear 
systems in~\cite{DeePC}. Data-driven control algorithms that use disturbance observers and maintain 
input constraints have also been treated, e.g., in~\cite{Ilya2014,Ilya2017}. Other recent data-driven 
control schemes based on linear systems theory include~\cite{nofo18,tafr19}. {\em However, none of 
these data-driven control techniques guarantee nonlinear stability and robustness in discrete-time.} 
Guaranteed robustness and nonlinear stability in the sense of \cite{lyapthes} is required for wider 
applicability and reliability of data-driven control approaches for systems with modeling uncertainties 
and unknowns. 

While most prior work has used continuous-time feedback for data-driven control, 
the framework developed here uses discrete-time H\"{o}lder-continuous nonlinear 
model-free estimation and control for tracking desired output trajectories. 
This framework lays the foundation for nonlinearly stable and robust data-driven control in discrete time, 
using novel methods to estimate the local input-output model and for tracking control. Our past 
research using H\"{o}lder-continuous finite-time stable control and estimation schemes in continuous
time have appeared in, e.g.,~\cite{bosa_ijc,sasi:2017:finite,ecc19,asjc19}. The finite-time stable  
unknown (disturbance) and output observers can be designed to converge in 
a finite time period that is smaller than the settling time of the controller. For tracking a desired 
output or state trajectory, a H\"{o}lder-continuous nonlinear finite-time stable (FTS) tracking control 
scheme was used in~\cite{bosa_ijc,sasi:2017:finite}. These FTS schemes are based on the Lyapunov 
analysis in~\cite{bhatFT}, using H\"{o}lder-continuous Lyapunov functions. More recently, we 
developed a discrete-time version of these FTS schemes and used it for tracking 
control (~\cite{dfts-cdc19}). Here, we extend this basic discrete-time analysis to provide finite-time 
stable learning of unknown dynamics in real time. {\em The resulting framework for data-driven 
control provides guaranteed nonlinear stability of the overall feedback system without requiring 
high frequencies for measurement or control}. The overall emphasis in our approach is towards 
{\em guaranteeing nonlinear stability and robustness of the feedback system} using FTS 
learning of unknown dynamics. 

We develop this framework in discrete time as that makes it easier to implement numerically and 
experimentally. 
In the first part of our nonlinear model-free control framework, basic theoretical results on nonlinear 
H\"{o}lder-continuous finite-time stabilization in discrete time are developed; some of these results 
have not appeared before. In the second part of our framework, these results are used to design two 
nonlinearly stable and robust observers that predict the unknown model describing the input-output 
dynamics locally in state space and time, based on prior observed input-output behavior. This is a 
critical component of our framework, as these observers work as disturbance observers that are used 
to compensate for the unknown dynamics and ensure nonlinear stability of the overall feedback loop. 
In the third part of this framework, we design a nonlinearly stable, trajectory tracking 
control scheme to track a desired output trajectory. This nonlinearly stable control scheme 
is designed to ensure convergence of the output tracking error to zero in finite time, if the disturbance 
estimates are perfect. It is shown to be nonlinearly stable and robust to bounded errors in the 
disturbance estimates. Further, if the disturbance estimates are obtained from the disturbance 
observers designed in the second part of our framework, we show that the tracking errors are 
guaranteed to converge to a neighborhood of zero tracking errors.

The remainder of this paper is organized as follows. The mathematical formulation and assumptions  
on discrete-time nonlinear systems along with the theory of H\"{o}lder-continuous finite-time 
stability in discrete time, are developed in Section \ref{sec: nonsysfts}. This is a novel contribution of 
this paper.
In section \ref{sec: ulm}, two finite-time stable uncertainty observers are designed to estimate the unknown 
local model relating the inputs and outputs of the nonlinear system. This model relates the observed outputs 
of the system to the given inputs, and accounts for the combined effects of unknown internal states and 
unknown external inputs (disturbances). This is another novel contribution of this paper.
Two model-free control laws for output tracking are given in Section \ref{sec: mfccon}, and this is yet 
another novel contribution of this paper. 
These discrete time control laws make the feedback system converge to the desired output trajectory in a 
finite-time stable manner if the estimated model converges exactly to the true input-output dynamics. 
Section \ref{sec: numres} provides numerical simulation results of applying this nonlinear model-free control 
framework to output trajectory tracking of a two-degree of freedom mechanical system. The model of the  
system is assumed to be unknown for purposes of control design, and it includes nonlinear 
friction terms affecting the dynamics of both degrees of freedom. The results of this numerical 
simulation agree with the analytical stability and robustness properties of this control framework. Finally, 
section \ref{sec: conc} provides a summary of the main results and ends with planned future work.

\section{Nonlinear system assumptions and finite-time stability in discrete time}\label{sec: nonsysfts}
\vspace*{-1mm}
In the first part of this section, we lay out the assumptions on discrete nonlinear systems for 
which the proposed framework of data-driven control are applicable. In the second part of this section, we 
give a result on finite-time stability of such discrete nonlinear systems. 

\subsection{Assumptions on discrete nonlinear system for our framework of data-driven control}\label{ssec: nsysa}
\vspace*{-1mm}
Consider a discrete nonlinear system with $m$ inputs, $n$ outputs, and $l$ unknown internal parameters 
or states. We denote by $\bR$ the set of all real numbers, $\bR^+$ the set of all positive numbers, and 
$\bR_0^+$ the set of all non-negative numbers.
The notation $(\cdot)_k=(\cdot) (t_k)$ denotes the value of a time-varying quantity at sampling instant 
$t_k\in\bR_0^+$, with $u_k\in\bR^m$ as the input vector comprising control inputs, $y_k\in\bR^n$ is the 
output vector consisting of measured output variables, and $z_k\in\bR^{l}$ is the vector of unknown 
(unobservable) internal parameters and states. Here $k\in \bW=\{ 0,1,2,\ldots\}$ and $\bW$ is the index 
set of whole numbers including $0$. 
%

We use the superscript $(\mu)$ to denote the $\mu$th order finite difference of a quantity in discrete time. 
The forward difference defined by
\be y_k^{(\mu)} := y_{k+1}^{(\mu-1)}-y_k^{(\mu-1)} \mbox{ with } y_k^{(0)}=y_k \label{fordif} \ee
is used here, because of its simplicity and applicability. Let $\nu$ be the relative degree of the input-output 
system. The unknown discrete system can be expressed as 
\be y_k^{(\nu)}= \varpi (y_k,y_k^{(1)},\ldots,y_k^{(\nu-1)},z_k,u_k,t_k), \label{discmod} \ee
where $\varpi$ is unknown, $\nu$ is known and $y_k^{(\mu)}$ is the $\mu$th order finite difference 
defined by eq. \eqref{fordif}. Using  eq. \eqref{fordif}, we can alternately represent the system 
\eqref{discmod} as
\be y_{k+\nu}= \varphi (y_k,y_{k+1},\ldots,y_{k+\nu -1},z_k,u_k,t_k). \label{rdiscsys} \ee 
Note that $\varpi,\varphi: (\bR^n)^{\nu}\times\bR^l\times\bR^m\times\bR_0^+\to\bR^n$ are possibly 
time-varying but unknown or imperfectly known. The control inputs $u_k$ are then designed so as to track 
a desired output trajectory $y^d_k=y^d(t_k)$. The following basic assumptions are made in order to have 
a tractable discrete-time output tracking control problem.  

\begin{assumption} \label{assump1}
The discrete-time nonlinear system given by \eqref{discmod} (alternately \eqref{rdiscsys}) has a 
$\varpi(\cdots)$ (alternately $\varphi(\cdots)$) that is unknown but Lipshcitz continuous in all arguments. 
Further, $\varphi(\cdots)$ can be represented as:
\begin{align} 
\begin{split}
&\varphi(y_k,\ldots,y_{k+\nu-1},z_k,u_k,t_k)= \sF_k +\sG_k u_k, \mbox{ where} \\
&\sF_k = \sF (y_k,y_{k+1},\ldots,y_{k+\nu -1},z_k,u_k,t_k) \mbox{and } \\
&\sG_k= \sG  (y_k,y_{k+1},\ldots,y_{k+\nu -1},z_k,u_k,t_k),
\end{split}\label{truedyn} 
\end{align}
where $\sF: \bR^{n\nu}\times\bR^l\times\bR^m\times\bR_0^+ \to\bR^n$ and $\sG: \bR^{n\nu}
\times\times\bR^m\times\bR_0^+ \to\bR^{n\times m}$ are Lipschitz continuous in their arguments. 
\end{assumption}
\begin{assumption} \label{assump2}
The matrix $\sG_k\in\bR^{n\times m}$ is full ranked with $m\ge n$, i.e., there are at least as many inputs 
as there are outputs. Therefore, the outputs $y_k\in\bR^n$ can be controlled by the inputs $u_k\in\bR^m$ at 
all time instants $t_k$. 
\end{assumption}
\begin{assumption} \label{assump3}
The desired output trajectory $y^d_k:= y^d(t_k)$ to be tracked by the discrete system \eqref{discmod} 
(alternately \eqref{rdiscsys}), is obtained by sampling a continuous and $\nu$ times differentiable trajectory 
in time, $y^d(t)$, with bounded time derivatives up to order $\nu$. 
\end{assumption}
Let $\chi_k=(y_k,y_{k+1},\ldots,y_{k+\nu -1},z_k,u_k,t_k)$ denote the vector of variables on which the 
system \eqref{rdiscsys} depends. Then Assumption \ref{assump1} implies that:
\be \| y_{k+\nu+1}- y_{k+\nu}\| = L \|\chi_{k+1}-\chi_k\|, \label{Lipsch} \ee
where $L$ is the Lipschitz constant. 
When the relative degree $\nu$ is unknown, $\nu$ can be identified using known techniques 
(e.g.,~\cite{heas93, rhmo98}), or a sufficiently high order may be assumed for model-free control. 

In practice, outputs are measured by sensors that usually introduce noise, modeled by:
\be y^m (t_k)= y^m_k= y_k+ \eta_k, \label{measnoise} \ee
where $\eta_k\in\bR^l$ is a vector of additive noise. A nonlinearly stable filtering scheme can be 
used to filter out this measurement noise. In particular, nonlinear finite-time stable observers 
can filter out noise and provide rapid convergence, as shown recently in~\cite{dftsobs19}.

\subsection{Finite-time stabilization in discrete time using H\"{o}lder continuous feedback}\label{ssec: ftsdisc}
\vspace*{-1mm}

This section gives a basic result on finite-time stability for discrete time systems that leads to a 
H\"{o}lder-continuous system. This result has the following benefits in our framework for model-free control: 
(1) the added robustness of finite-time stability compared to asymptotic stability for nonlinear systems when
faced with the same intermittent or persistent disturbances~\cite{bhatFT, bosa_ijc}; and (2) convergence to 
zero errors in finite time makes it easier to analyze stability and robustness of the feedback system after 
separate observer and controller designs. 
This basic result 
has recently appeared in \cite{dftsobs19, dfts-cdc19}. Here we give the same result with a simpler 
mathematical proof, followed by a new result showing H\"{o}lder continuity of the system. 

\begin{lemma}\label{discFTSlem}
Consider a discrete-time system with outputs $s_k\in\bR^p$. Let $V:\bR^p \to \bR$ be a positive definite,  
decrescent and radially unbounded (Lyapunov) function (\cite{vid02}) of the outputs, and denote $V_k := 
V (s_k)$. Let $\alpha$ be a constant in the open interval $]0,1[$. 
Denote $\gamma_k:= \gamma (V_k)$ where $\gamma: \bR_0^+ \to\bR_0^+$ is a positive definite 
function of $V_k$ that satisfies the condition that there exists $\varepsilon\in\bR^+$ such that:
\be 
\gamma_k= \gamma (V_k) \ge \eta:=\varepsilon^{1-\alpha}\, \mbox{ for all }\, V_k \ge \varepsilon.  
\label{gammacond} \ee
Then, if $V_k$ satisfies the relation: 
\be V_{k+1}- V_k \le -\gamma_k V_k^\alpha, \label{discFTS} \ee
the discrete system is stable at $s=0$ and $s_k$ converges to $s=0$ for $k> N$, where  
$N\in\bW$ is finite.
\end{lemma}
\begin{proof}
Clearly, inequality \eqref{discFTS} is a sufficient condition for asymptotic stability of the zero output 
$s=0$, as it ensures that the difference $V_{k+1}- V_k$ along output trajectories of the discrete-time system 
is negative definite. The right-hand side of the inequality \eqref{discFTS} is zero if and only if $V_k=0$, as 
$\gamma(\cdot)$ is a positive definite function of $V_k$. Therefore the sequence $\{V_k\}$ is a monotonically 
decreasing sequence in $\bR_0^+$. Combining condition \eqref{gammacond} with inequality \eqref{discFTS}, 
we obtain:
\begin{align} 
V_{k+1} \le V_k -\gamma_k V_k^\alpha \le V_k- \eta V_k^{\alpha}. \label{Vkp1in} 
\end{align}

The rest of the proof uses contradiction to arrive at the given result. Note that $V_k$ decreases at least 
as fast as the far right-hand side of the inequality \eqref{Vkp1in}. 
Now define 
\be c_k := V_k/\varepsilon, \label{ckdef} \ee
where $V_k= V(s_k)$ is the value of the Lyapunov function at time $t_k$. 
Substituting this expression for $V_k$ in inequality \eqref{discFTS}, we obtain:
\be \begin{split} 
V_{k+1} - c_k \eta^{\frac1{1-\alpha}} \le -\eta c_k^\alpha  \big(\eta\big)^{\frac{\alpha}{1-\alpha}} 
= -c_k^\alpha \eta^{\frac1{1-\alpha}} \\
\Rightarrow V_{k+1} \le (c_k -c_k^\alpha) \eta^{\frac1{1-\alpha}}= (c_k -c_k^\alpha)\varepsilon.
\end{split} \label{V1eqn}
\ee
From  eq. \eqref{ckdef} $c_{k+1}:= V_{k+1}/\varepsilon$, and from inequality \eqref{V1eqn}:
\be c_{k+1}\le a_k\, \mbox{ where }\, a_k:= c_k -c_k^\alpha. \label {ckp1def} \ee
Next we consider two cases: (1) $c_k\le 1$; and (2) $c_k> 1$. In the first case, $c_{k+1}\le a_k\le 0$ from 
expression \eqref{ckp1def}, as $\alpha\in]0,1[$. Therefore, from \eqref{V1eqn} we get $V_{k+1}\le 0$. 
This leads to a contradiction because $V_{k+1}=V(s_{k+1})$ is a positive definite function of $s_{k+1}$. 
Therefore $V_{k+1}=0$ and as a result $s_{k+1}=0$, which leads to convergence of the output $s_j$ to
the zero vector for all $j>N=k$.

In the second case when $c_k>1$, we see from eq. \eqref{ckp1def} that $c_{k+1}\le a_k$ and $a_k>0$. 
Now we repeat the  previous few steps defined by expressions \eqref{ckdef}-\eqref{ckp1def} by replacing 
$k\leftarrow k+1$.
As long as $c_j>1 \Leftrightarrow a_j>0$, this process can be continued to obtain monotonically 
decreasing sequences of positive real numbers $\{V_j\}$ and $\{c_j\}$ as follows:
\be V_{j+1} =c_{j+1} \varepsilon\, \mbox{ where }\, c_{j+1}\le a_j := c_j -c_j^\alpha, \; j\in\bW. 
\label{Vkp1eq} \ee
Clearly, because $\alpha\in ]0,1[$, the sequence $\{c_j\}$ defined by eq. \eqref{Vkp1eq} is monotonically 
decreasing and $c_j\in\bR^+$ for $c_j\ge 1$. Let $N\in\bW$ ($N>k$) be 
the smallest finite whole number such that $c_N\le 1$ (and correspondingly $a_N\le 0$) in this sequence. 
Therefore $V_N=c_N\varepsilon\le \varepsilon = \eta^{\frac1{1-\alpha}}$ and $a_N= c_N -c_N^\alpha\le 0$. 
Then from eq. \eqref{Vkp1eq} we have:
\[ V_{N+1}= c_{N+1} \varepsilon \le a_N \varepsilon \le 0. \]
As before, this leads to a contradiction as $V_{N+1}=V(s_{N+1})$ cannot be negative by definition; it 
has to be zero. Consequently, from the inequality \eqref{discFTS}, we see that $V_j=0$ for $j\ge N+1$. 
As a result, $s_j$ converges to $s=0$ for $j>N$. 
\end{proof}


\begin{remark}\label{rem2}
Note that $\varepsilon\in\bR^+$ is not given by this result; the result merely states that if such a positive 
$\varepsilon$ exists that satisfies condition \eqref{gammacond}, then finite-time stability of $s=0$ in discrete time 
is guaranteed. In fact, inequality \eqref{gammacond} is easy to satisfy and holds true for all positive definite 
class-$\cK$ functions $\gamma_k= \gamma (V_k)$. In particular, positive definite sigmoid functions are a 
good choice for $\gamma (V_k)$. 
\end{remark}

The result below shows the H\"{o}lder continuity of a Lyapunov function that satisfies condition 
\eqref{discFTS} of Lemma \ref{discFTSlem}.

\begin{theorem}\label{thm0}
A discrete-time Lyapunov function that satisfies inequality \eqref{discFTS} is H\"{o}lder continuous in 
discrete time with exponent $\frac1{1-\alpha}$.
\end{theorem}
\begin{proof} 
From Lemma \ref{discFTSlem} and its proof, it is clear that $V_j> 0$ as long as $V_j>\varepsilon$. Further, 
it is also clear that $V_j=0$ if $V_k\le\varepsilon$ for any $k<j$. Now consider indices $i,j\in\bW$ such that 
$V_i, V_j >\varepsilon$ (i.e., $i,j <N$, where $N$ is as defined in the proof of Lemma \ref{discFTSlem}). 
Therefore, we have:
\begin{align} 
V_{i+1}- V_0 &= V_{i+1}-V_i +V_i-V_{i-1}+ \ldots+ V_1-V_0 \nn \\
&\le -\eta (V_i^\alpha+\ldots+V_0^\alpha) < -(i+1)\eta V_{i+1}^\alpha \nn \\
\Rightarrow & V_{i+1}^{1-\alpha}- \frac{V_0}{V_{i+1}^\alpha} < -(i+1)\eta, \label{Vip1eq}
\end{align}
where we used the fact that $V_{i+1}<V_k$ for $k<i+1$. Re-arranging inequality \eqref{Vip1eq} and 
using $V_{i+1}>\varepsilon$, we have:
\begin{align}  
\begin{split}
V_{i+1} &< \Big( \frac{V_0}{V_{i+1}^\alpha} -(i+1)\eta\Big)^\frac1{1-\alpha}  \\ 
&< \big(\upsilon -(i+1)\eta\big)^\frac1{1-\alpha}, \mbox{ where } \upsilon =\frac{V_0}{\varepsilon^\alpha}. 
\end{split} \label{Vip1Hold} 
\end{align}
A similar inequality holds if $i$ is replaced by $j$ in the expression above. Therefore, we conclude that:
\begin{align}
\frac{|V_{i+1}-V_{j+1}|}{|i-j|^\frac1{1-\alpha}} < \varepsilon + o(2) \label{Holderineq}
\end{align}
using the fact that $\varepsilon=\eta^\frac1{1-\alpha}$, where $o(2)$ denotes second and higher order 
terms in $|i-j|$. Clearly, the above H\"{o}lder inequality holds trivially if either $V_{i+1}$ or $V_{j+1}$ 
or both are zero (i.e., $i\ge N$ and/or $j\ge N$). Therefore, the sequence $\{V_k\}\subset\bR^+_0$ 
is H\"{o}lder-continuous in discrete time as given in the statement.
\end{proof}

\begin{remark}
The statement of Theorem \ref{thm0} for the Lyapunov function holds if the discrete time system is replaced 
by a continuous time system satisfying 
\[ \dot V \le -\eta V^\alpha, \]
for constant $\eta$ as was shown in Theorem 4.3 of \cite{bhatFT}. However, this result for discrete-time 
systems has not appeared in prior publications. 
\end{remark}

The condition given in Lemma \ref{discFTSlem} can be satisfied easily for quadratic Lyapunov functions, 
as the main result of the following section, which is in the form of a finite-time stable 
disturbance observer in discrete time, shows.

\section{The Ultra-Local Model and Its Estimation}\label{sec: ulm}
\vspace*{-1mm}
In this section, we construct a control affine ultra-local model (ULM) in discrete time that models the 
unknown dynamics as a disturbance input to a control-affine system, and then estimates 
this disturbance input from past input-output data. We design first and second order discrete time nonlinear 
observers that estimate this disturbance input. Convergence to this disturbance input is achieved in finite time 
if this input is a constant vector; otherwise, these observers are shown to be robust to bounded rates of 
change of this disturbance input. 
This control-affine ULM along with the disturbance observer is used in Section \ref{sec: mfccon} to construct 
an output feedback control scheme to track a desired output trajectory. 

\subsection{Ultra-Local Model for Unknown System}\label{ssec: ulmd}
\vspace*{-1mm}
The idea of an ultra-local model that is local in output (or state) space and in time, was proposed in 
the model-free control approach for SISO systems in~\cite{FlJo13}. In this work, we generalize this 
concept to a discrete time nonlinear system with unknown dynamics of the form given by eq. 
\eqref{rdiscsys}, as follows:
\be y_{k+\nu} = \cF_k +\cG_k  u_k,\, \mbox{ where }\, \cF_k\in\bR^n,\; u_k\in\bR^m, 
\label{ultramodel} \ee
and $\cG_k \in \bR^{n\times m}$ is a full rank matrix that is designed or selected appropriately, as part of 
the controller design. Note that the $\cF_k$ and $\cG_k$ so obtained may not be unique and in particular, 
may not be equal to the $\sF_k$ and $\sG_k$ respectively in eq. \eqref{truedyn}, which is  
part of the Assumption \ref{assump1} for the system \eqref{rdiscsys}. However, without any 
loss of generality, they can always be represented such that:
\begin{align}
\begin{split}
\cF_k = \cF (y_k,y_{k+1},\ldots,y_{k+\nu -1},z_k,u_k,t_k), \mbox{and } \\
\cG_k= \cG  (y_k,y_{k+1},\ldots,y_{k+\nu -1},z_k,u_k,t_k).
\end{split}\label{relmodel} 
\end{align}
Further, based on Assumption \ref{assump1}, we can assume that:
\begin{align}
\begin{split}
\|\cF_{k+1}-\cF_k\| \le L_F \|\chi_{k+1}- \chi_k\|, \\
\|\cG_{k+1}-\cG_k\| \le L_G \|\chi_{k+1}- \chi_k\|,
\end{split} \label{LipsconsFG}
\end{align}
where $L_F$ and $L_G$ are Lipschitz constants and $\chi_k$ is as defined before eq. \eqref{Lipsch}.

The approach given here is centered around provable guarantees on nonlinear stability and robustness to 
the unknown dynamics. To do this in an effective manner, the unknown input-output dynamics, 
captured by $\cF_k\in\bR^n$ in eq. \eqref{ultramodel}, should be estimated in a stable and robust 
manner. We therefore consider the following problem.  
\begin{problem}\label{prob1}
Consider the unknown nonlinear system \eqref{rdiscsys} satisfying Assumptions \ref{assump1}, \ref{assump2} 
and \ref{assump3}, with discrete control inputs $u_k:=u(t_k)\in\bR^m$ provided at 
sampling times $t_k$. Given the discrete time ultra-local model \eqref{ultramodel} of the input-output  
dynamics with unknown $\cF_k$, estimate $\cF_k$ from past input-output history and design a feedback 
control scheme to track the desired output trajectory $y^d_k:= y^d(t_k)$ in a nonlinearly stable manner.  
\end{problem} 

Note that as per Assumption \ref{assump2}, the system is input-output controllable. 
In the following subsections of this section, we design two nonlinear observers to estimate $\cF_k$ for later 
use output feedback tracking control. These schemes (in isolation) can also be used to identify this unknown 
dynamics using known (feedforward) control inputs $u_k$ and influence matrix $\cG_k$. Note that this influence 
matrix $\cG_k$ can also be designed to satisfy known control bounds, although we do not do that here. 


\subsection{Estimation of Unknown Input-Output Dynamics Using a First Order Observer}\label{ssec: estF1st}
Note that the model given by \eqref{ultramodel} is a generalization of the ultra-local model of~\cite{FlJo13}, 
where $\cG_k$ was a constant scalar and only single-input single-output (SISO) systems were considered. 
Here, we provide a first-order observer for this unknown dynamics, i.e., $\cF_k$ in eq. \eqref{ultramodel}. 
The idea here is to use the finite-time stable output observer design outlined in the previous section in 
conjunction with a first-order hold to estimate the unknown dynamics expressed by $\cF_k$ in eq. 
\eqref{ultramodel} based on past input-output history. Note that the control law for $u_{k}$ cannot 
use feedback of $\cF_{k}$ which is unknown due to causality; but it can use an estimate of $\cF_k$, 
denoted $\hat\cF_k$ here, based on past information on $\cF_j$ for $j\in\{0,\ldots,k-1\}$. This approach 
is very different from the approach used in the iPID framework for continuous-time model-free control, 
which is based on numerical differentiation and estimation of derivatives from noisy signals in the Laplace 
domain (see~\cite{Mboup2009}).

Define the error in estimating $\cF_k$ as follows:
\be e^\cF_k := \hat\cF_k- \cF_k. \label{esterrF} \ee
The following result gives a first order nonlinearly stable observer for the unknown dynamics $\cF_k$. 
\begin{theorem}\label{prop1st}
Let $e^\cF_k$ be as defined by eq. \eqref{esterrF}, and let $r \in ]1,2[$ and $\lambda >0$ be constants. 
Let the first order finite difference of the unknown dynamics $\cF_k$, given by 
\be \Delta\cF_k:= \cF_k^{(1)}= \cF_{k+1}- \cF_k, \label{DelF} \ee
be bounded as in the first of eqs. \eqref{LipsconsFG}. Let the control influence matrix $\cG_k$ be designed 
such that its first order finite difference is bounded as in the second of eqs. \eqref{LipsconsFG}. Consider 
the nonlinear observer for $\cF_k$ given by:
\begin{align}
\begin{split}
&\hat\cF_{k+1} = \cD (e^\cF_k) e^\cF_k + \cF_k \mbox{ given } \hat\cF_0, \\
&\mbox{where }  \cD (e^\cF_k )= \frac{\big((e^\cF_k)\T e^\cF_k\big)^{1-1/r} -\lambda}{\big((e^\cF_k)\T 
e^\cF_k\big)^{1-1/r} +\lambda}, 
\end{split} \label{Fest1st}
\end{align}
and $\cF_k = y_{k+\nu}-\cG_k u_k$ according to the ultra-local model \eqref{ultramodel}.
This observer leads to finite time stable convergence of the estimation error vector $e^\cF_k \in\bR^l$ to a 
bounded neighborhood of $0\in\bR^n$, where bounds on this neighborhood can be obtained from bounds on 
$\Delta\cF_k$.  
\end{theorem}
\begin{proof}
The proof of this result begins by showing that if 
\be e^\cF_{k+1}= \cD (e^\cF_k) e^\cF_k, \label{FestFTS} \ee
where $\cD (e^\cF_k)$ is as defined by eq. \eqref{Fest1st}, then $e^\cF_k$ converges to zero in a 
finite-time stable (FTS) manner. This can be shown by defining the discrete-time Lyapunov function
\be V^\cF_k := (e^\cF_k)\T e^\cF_k. \label{VcFdef} \ee
Taking the discrete time difference of this Lyapunov function, we get
\begin{align} 
\begin{split}
&V^\cF_{k+1}- V^\cF_k = -\gamma^\cF_k (V^\cF_k)^{1/r} \\
&\mbox{where } \gamma^\cF_k= \big( 1 -(\cD (e^\cF_k))^2\big)(V^\cF_k)^{(1-1/r)}.
\end{split} \label{diffVcF} 
\end{align}
Note that $\cD(e^\cF_k)$ is a monotonically decreasing function of $\|e^\cF_k\|\in\bR^+_0$, taking values 
in the range $[-1,1)$, so $\gamma^\cF_k$ is a positive definite function of $V^\cF_k= \|e^\cF_k\|^2$.
Substituting $(e^\cF_k)\T e^\cF_k= V^\cF_k$ into the expression for $\mathcal{D}(e^\cF_k)$ to evaluate  
$\gamma^\cF_k$ in eq. \eqref{diffVcF}, we express $\gamma^\cF_k$ as a function of $V^\cF_k$: 
\be \gamma^\cF_k = 4\lambda \frac{ (V^\cF_k)^{2-2/r}}{\big( (V^\cF_k)^{1-1/r} +\lambda\big)^2}. 
\label{gamlk} \ee
Clearly, eq. \eqref{gamlk} shows that $\gamma^\cF_k:=\gamma (V^\cF_k)$ is a class-$\cK$ function of 
$V^\cF_k$. Further, it can be verified that: 
\[ V^\cF_N \le \lambda^{\frac1{1-1/r}}\, \Longleftrightarrow\, \gamma^\cF_N\le \lambda.  \]
These two facts together allow as to conclude that this $\gamma^\cF_k$ (taking the role of $\gamma_k$ in 
Lemma \ref{discFTSlem}) satisfies the sufficient condition \eqref{gammacond} for finite-time stability of 
$e^\cF_k$, with a value of $\epsilon= \lambda^{\frac1{1-1/r}}$. 
Using the definition of $e^\cF_k$ given by eq. \eqref{esterrF} and the relation \eqref{FestFTS}, 
one obtains the following discrete time observer for $\hat\cF_k$:
\be \hat\cF_{k+1}= \cD (e^\cF_k) e^\cF_k + \cF_{k+1}. \label{idealFest} \ee
The above expression leads to a finite-time stable observer for the unknown dynamics (disturbance 
input) $\cF_k$ that ensures that the estimation error $e^\cF_k$ converges to zero for $k>N$ where 
$N\in\bW$ is finite.

However, as mentioned earlier, $\cF_{k+1}$ is not available at time $t_{k+1}$ due to causality; therefore, it 
needs to be replaced by a known quantity. This first order observer design given by eq. \eqref{Fest1st} 
replaces $\cF_{k+1}$ in eq. \eqref{idealFest} with $\cF_k=y_{k+\nu}-\cG_k u_k$, which is known from the 
measured output and applied input from the previous sampling instant. As a result, the estimation 
error $e^\cF_k$ evolves according to:
\begin{align}
\begin{split} 
 &e^\cF_{k+1} := \hat\cF_{k+1}- \cF_{k+1} = \cD (e^\cF_k) e^\cF_k + \cF_k- \cF_{k+1} \\
 &= \cD (e^\cF_k) e^\cF_k - \Delta\cF_k,\, \mbox{ where } \Delta\cF_k=\cF_{k+1}-\cF_k. 
\end{split}\label{err1st}
\end{align}
Therefore this observer is a first order perturbation of the ideal FTS observer design for $\cF_k$ as 
given by eq. \eqref{idealFest}, with the perturbation coming from the first order difference term 
$\Delta\cF_k$. Due to the FTS behavior of this ideal observer for $\cF_k$, the first order observer 
design of eq. \eqref{Fest1st} will converge to a neighborhood of $e^\cF_k=0$ for finite $k$, where the 
size of this neighborhood depends on bounds on $\Delta\cF_k$. As $\cF_k$ is Lipschitz continuous, 
so $\|\Delta\cF_k\|$ is bounded according to the first of eqs. \eqref{LipsconsFG}. Clearly, the smaller the 
bounds on $\Delta\cF_k$, the smaller the neighborhood of $e^\cF_k=0$ that this observer will converge 
to within finite time. 
\end{proof}

In the following result, the observer given by eq. \eqref{Fest1st} is analytically shown to be robust to known 
bounds on the norm of the first difference $\Delta\cF_k$ in the unknown $\cF_k$. 
\begin{theorem}\label{thecFbnd}
Consider the observer law \eqref{Fest1st} for the unknown $\cF_k$ in the ultra-local model 
\eqref{ultramodel} modeling the unknown system \eqref{rdiscsys}. Let the first order difference 
$\Delta\cF_k$ defined by \eqref{DelF} be bounded according to:
\be \|\Delta\cF_k\| \le B^\cF, \label{DelFbd} \ee
where $B^\cF\in\bR^+$ is a known constant. Then the observer (estimation) error 
$e^\cF_k$ is guaranteed to converge to the neighborhood given by:
\be \cN^\cF :=\{ e^\cF_k\in\bR^n\, :\, \rho (e^\cF_k) \|e^\cF_k\|\le B^\cF  \} \label{ecFnbhd} \ee
for finite $k>N$, $N\in\bW$, where 
\begin{align} 
\begin{split}
& \rho (e_k^\cF) := \frac{\zeta (e^\cF_k)}{1-\sqrt{1-\zeta (e^\cF_k)}} \mbox{ and} \\ 
& \zeta (e^\cF_k):=1-\big(\cD(e^\cF_k)\big)^2= \frac{\gamma_k^\cF}{\big((e^\cF_k)\T e^\cF_k\big)^{1-1/r}}. 
\end{split} \label{rhodef}
\end{align}
\end{theorem}
\begin{proof}
Using the Lyapunov function defined by eq. \eqref{VcFdef} and the observer eq. \eqref{Fest1st}, we 
obtain:
\begin{align*}
& V^\cF_{k+1}- V^\cF_k = (e_{k+1}^\cF +e^\cF_k)\T  (e_{k+1}^\cF -e^\cF_k) \\
&=\big( (\cD (e^\cF_k))^2 -1\big) (e^\cF_k)\T e^\cF_k -2\cD(e^\cF_k)\Delta F_k\T e^\cF_k \\ 
&+\Delta F_k\T \Delta F_k. 
\end{align*}  
Taking into account the bound \eqref{DelFbd} and the expression for $\gamma^\cF_k$ given by eq. 
\eqref{diffVcF}, we get an upper bound on the first difference of this Lyapunov function as follows:
\begin{align} &V^\cF_{k+1}- V^\cF_k \le -\gamma^\cF_k (V^\cF_k)^{1/r} +2|\cD(e^\cF_k)| 
B^\cF \|e^\cF_k\| +(B^\cF)^2 \nn \\
&=-\zeta (e^\cF_k) \|e^\cF_k\|^2 +2 |\cD(e^\cF_k)| B^\cF \|e^\cF_k\| +(B^\cF)^2, \label{delVFbds} \end{align}
using the definition of the Lyapunov function \eqref{VcFdef} and eq. \eqref{rhodef} for $\zeta (e^\cF_k)$ 
in the last step. 
For large enough initial (transient) $\|e^\cF_k\|$, the right-hand side of the inequality 
\eqref{delVFbds} is negative, and we get:
\[ \zeta (e^\cF_k) \|e^\cF_k\|^2 +2|\cD(e^\cF_k)| B^\cF \|e^\cF_k\| -(B^\cF)^2 >0, \]
which can be solved as a quadratic inequality expression in $\|e^\cF_k\|$ with coefficients that depend on 
$e^\cF_k$. Noting that $\zeta (e^\cF_k)= 1-(\cD (e^\cF_k))^2$, this leads to the condition:
\be \zeta (e^\cF_k) \|e^\cF_k\| > B^\cF \Big[1-\sqrt{1-\zeta(e^\cF_k)}\Big], \label{ecFbnd} \ee
for real positive solutions of $\|e^\cF_k\|$ for which $V^\cF_{k+1}-V^\cF_k<0$ is guaranteed. The discrete 
Lyapunov function $V^\cF_k$ is monotonically decreasing for such $\|e^\cF_k\|$ large enough to satisfy 
inequality \eqref{ecFbnd}, which it will until a finite value of $k$, say $k=N$. Therefore, the observer error 
$e^\cF_k$ is guaranteed to converge to the neghborhood $\cN^\cF$ of $0\in\bR^n$ given by equations 
\eqref{ecFnbhd}-\eqref{rhodef} and will remain in this neighborhood (which is positively invariant) for $k>N$.
\end{proof}

\begin{remark}
This first oder observer can become unstable if $\Delta\cF_k$ escapes (becomes unbounded) in 
finite time at a rate faster than that dictated by the design of $\cD (e^\cF_k)$. However, the Lipschitz 
continuity condition imposd by eqs. \eqref{LipsconsFG} ensures that this cannot happen. The H\"{o}lder 
continuity of the feedback system, as given by Theorem \ref{thm0}, guarantees robustness to Lipschitz 
continuous disturbances according to Theorem \ref{thecFbnd}.
\end{remark}

\begin{remark}
Note that the bound on $\|e^\cF_k\|$ given by the sufficient condition in Theorem \ref{thecFbnd} is 
conservative. Moreover, as $\Delta\cF_k$ is bounded according to the Lipschitz condition \eqref{LipsconsFG}, 
the bound on it given by \eqref{DelFbd} is already likely to be conservative for small changes in states between 
output measurements. Therefore, this sufficient condition may be said to be ``doubly conservative."
\end{remark}


\subsection{Estimation of Unknown Input-Output Dynamics Using a Second Order 
Observer}\label{ssec: estF2nd}
\vspace*{-1mm}
In this subsection, we design a second order observer for $\cF_k$ based on the developments in the 
previous subsection. To start the design process, we reverse eq. \eqref{DelF} toobtain:
\be \cF_{k+1}=\cF_k + \Delta\cF_k. \label{intmodcF} \ee
The second order observer design is based on the above expression, as follows:
\be \hat\cF_{k+1}= \hat\cF_k +\Delta\hat\cF_k, \label{est2nd} \ee 
where $\Delta\hat\cF_k$ is the estimate of $\Delta\cF_k$. In addition, define the error in estimating 
$\Delta\cF_k$ as follows:
\be e^\Delta_k := \Delta\hat\cF_k - \Delta\cF_k. \label{erDelcF} \ee
The following result gives the second order observer designed to estimate $\cF_k$ 
and $\Delta\cF_k$. 

\begin{theorem}\label{prop2nd}
Let $e^\Delta_k$ be as defined by eq. \eqref{erDelcF}, and $e^\cF_k$, $r\in ]1,2[$ and $\lambda >0$ 
be as defined in Theorem \ref{prop1st}. Further, let $\cD (\cdot)$ be as defined by eq. \eqref{Fest1st} 
in Theorem \ref{prop1st}, and let the second order finite-time difference given by: 
\be \Delta^2\cF_k:= \cF_{k-1}^{(2)}= \cF_{k+1}- 2\cF_k +\cF_{k-1} \label{Del2F} \ee
be bounded as obtained from the first of eqs. \eqref{LipsconsFG}. Let the control influence matrix $\cG_k$ 
be designed such that its first order finite difference is bounded as in the second of eqs. \eqref{LipsconsFG}. 
Consider the nonlinear observer given by:
\begin{align}
\begin{split}
&\hat\cF_{k+1} = \cD (e^\cF_k) e^\cF_k +\cF_k + \Delta\hat\cF_k \mbox{ given } \hat\cF_0, \\
&\mbox{where }  \Delta\hat\cF_k= \cD (e^\Delta_{k-1}) e^\Delta_{k-1} + \Delta\cF_{k-1},
\end{split} \label{Fest2nd}
\end{align}
and $\cF_k = y_{k+\nu}-\cG_k u_k$ according to the ultra-local model \eqref{ultramodel}. 
This observer leads to finite time stable convergence of the estimation errors $(e^\cF_k, e^\Delta_k) \in
\bR^n\times\bR^n$ to bounded neighborhoods of $(0,0)\in\bR^n\times\bR^n$, where these bounds can  
be obtained from bounds on $\Delta^2\cF_k$. 
\end{theorem}
\begin{proof}
The proof of this result starts by noting that the ideal FTS observer law for $\cF_k$ given by eq. 
\eqref{idealFest} can also be expressed as:
\be \hat\cF_{k+1}= \cD (e^\cF_k) e^\cF_k + \cF_k+ \Delta\cF_k, \label{idealFest2} \ee
because the last two terms on the right side of this expression add up to $\cF_{k+1}$. The second order 
observer law given by eq. \eqref{Fest2nd} is obtained by replacing $\Delta\cF_k$ on the RHS of eq. 
\eqref{idealFest2} with its estimate. The estimate $\Delta\hat\cF_k$ would converge to the true value 
$\Delta\cF_k$ in finite time, if it is updated according to the (ideal) observer law:
\be \Delta\hat\cF_k= \cD (e^\Delta_{k-1}) e^\Delta_{k-1} + \Delta\cF_k. \label{idealDFest} \ee
Note that this ideal observer for $\Delta\cF_k$ is of the same form as the ideal FTS observer law for 
$\cF_k$ given by eq. \eqref{idealFest}. Further, like the ideal observer \eqref{idealFest}, the observer eq. 
\eqref{idealDFest} is not practically implementable because $\Delta\cF_k$ is unknown at time $t_k$ 
(because $\cF_{k+1}$ is unknown). As we did with the first order observer in Theorem \ref{prop1st}, we 
replace $\Delta\cF_k$ in \eqref{idealDFest} with its previous value, assuming that the change in this  
quantity is small in the time interval $[t_{k-1},t_k]$. This leads to the following observer law for $\Delta\cF_k$: 
\be \Delta\hat\cF_k= \cD (e^\Delta_{k-1}) e^\Delta_{k-1} + \Delta\cF_{k-1}. \label{DFest} \ee
The resulting second order observer is therefore given by eqs. \eqref{Fest2nd}. To show that this is 
indeed second order, the evolution of the estimation error $e^\cF_k$ in discrete time is obtained as 
below:
\begin{align}
 e^\cF_{k+1} &:= \hat\cF_{k+1}- \cF_{k+1} = \cD (e^\cF_k) e^\cF_k  +\cD (e^\Delta_{k-1}) e^\Delta_{k-1} \nn \\ 
 &\; \; + \cF_k + \Delta\cF_{k-1} - \cF_{k+1} \label{err2nd} \\
 &= \cD (e^\cF_k) e^\cF_k +\cD (e^\Delta_{k-1}) e^\Delta_{k-1} - \Delta^2\cF_k, \nn
\end{align}
where $\Delta^2\cF_k$ is as defined by eq. \eqref{Del2F}. The last line in the above expression is 
obtained by substituting for $\Delta\cF_{k-1}$ in the previous line, using the definition of $\Delta\cF_k$ 
given by eq. \eqref{DelF}. The remainder of the proof of this result uses the same arguments as in the 
last part of the proof of Theorem \ref{prop1st}, with $\Delta\cF_k$ replaced by $\Delta^2\cF_k$ and 
$\Delta\cF_k$, $\Delta\cG_k$ bounded as in eqs. \eqref{LipsconsFG}. 
\end{proof}

\begin{remark}\label{remecFbd}
Sufficient conditions on the bounds on neighborhoods of $0\in\bR^n$ that the estimation errors $e^\cF_k, 
e^\Delta_k \in\bR^n$ converge to for the observer in Theorem \ref{prop2nd}, can be obtained in a manner 
similar to the bounds given by Theorem \ref{thecFbnd} for the estimation error obtained from the observer 
in Theorem \ref{prop1st}.
\end{remark}

\begin{remark}
It is clear from the constructive proofs of Theorems \ref{prop1st} and \ref{prop2nd} that higher order 
observers for $\cF_k$ may be constructed using a similar process. For example, a third order observer 
can be constructed by replacing $\Delta\cF_{k-1}$ in the second line of eq. \eqref{Fest2nd} with 
$\Delta\cF_{k-1}+ \Delta^2\hat\cF_{k-1}$ and finding an appropriate update law for $\Delta^2
\hat\cF_{k-1}$. Clearly, the added computational burden of higher order observers make them unattractive 
for implementation when the higher order differences of the discrete signal $\cF_k$ are known to be 
within reasonable bounds. In most situations, the uncertainty $\cF_k$ is Lipschitz continuous 
as given by eq. \eqref{LipsconsFG}, and bounds on $\Delta\cF_k$ and $\Delta^2\cF_k$ can be 
obtained; therefore, these low order observers are adequate.
\end{remark}

\begin{remark}\label{remcons}
Note that both the observers given by Theorems \ref{prop1st} and \ref{prop2nd} provide {\em 
exact} finite-time stable covergence of estimation errors $e^\cF_k$ (and $e^\Delta_k$ for Theorem 
\ref{prop2nd}) to zero if $\cF_k$ is constant. Therefore, the resulting feedback loop with either of these 
``disturbance" observers rejects constant disturbance signals $\cF_k$.
\end{remark}

\section{Model-free nonlinearly stable feedback tracking control}\label{sec: mfccon} 
\vspace*{-1mm}
In this section, we design tracking control laws for the control input $u_k$ 
from the output $y_k$, the desired output $y^d_k$, and the estimate of 
the ultra-local model $\hat\cF_k$ constructed from output measurement $y_k$ and past input-output 
history as described in Section \ref{sec: ulm}. 
This section provides two nonlinear model-free output feedback tracking control schemes that solve 
Problem \ref{prob1} in Section \ref{ssec: ulmd}. The control design process is based on Assumptions 
\ref{assump1}, \ref{assump2} and \ref{assump3} for the discrete nonlinear system \eqref{discmod}, 
and is designed to track a desired output trajectory for a system expressed by the 
ultra-local model \eqref{ultramodel}. The control designs 
given here may make use of either of the 
nonlinear observers for the ultra-local model 
given in Sections \ref{ssec: estF1st} and \ref{ssec: estF2nd}, but is independent of 
these observers designed in the previous section. They can be used in conjunction with other disturbance 
(or ultra-local model) observers that do the same task. 

\subsection{First Output Trajectory Tracking Control Law}\label{ssec: trajcon1}
\vspace*{-1mm}
Let $y^d :\bR\to \bR^n$ be a $\mC^\nu$ function that gives a desired output trajectory that is $\nu$ 
times differentiable according to Assumption \ref{assump3}, and denote $y^d_k:= y^d (t_k)$ for 
$\{t_k\}\subset \bR^+_0$. 
Considering Problem \ref{prob1}, define the output trajectory tracking error 
\be e^y_k= y_k - y^d_k\, \mbox{ where }\, y^d_k= y^d (t_k). \label{trackerrs} \ee
The first control law design presented here has the following objectives: (1) to ensure 
that the feedback system tracks the desired trajectory in a nonlinearly stable manner; and (2) to ensure 
that the tracking error is ultimately bounded by the same ultimate bounds that bound the observer 
error in the model estimate, $e^\cF_k$.

\begin{proposition}\label{contlaw1}
Consider an unknown input-output system described by the ultra-local model \eqref{ultramodel} with the 
control law:
\be \cG_k u_k= y^d_{k+\nu}- \hat\cF_k, \label{claw1} \ee
where the $y^d_i=y^d(t_i)$ describe a desired output trajectory for the time sequence $\{t_i\}\subset
\bR^+_0$ as in eq. \eqref{trackerrs}. The trajectory tracking error $e^y_k$ defined by \eqref{trackerrs} 
then satisfies 
\be e^y_{k+\nu}+ e^\cF_k =0, \label{errbds} \ee
where $e^\cF_k$ is the model estimation error defined by \eqref{esterrF}. In particular, if the model 
estimate $\hat\cF_k$ is given by the observer in Theorem \ref{prop1st} or Theorem \ref{prop2nd}, then 
$e^y_{k+\nu}$ converges to a bounded neighborhood of $0\in\bR^n$ in finite time (for finite $k$), where 
the bounds on this neighborhood are given by the same bounds that bound the estimation error $e^\cF_k$.
\end{proposition}
\begin{proof}
The short proof starts by noting that:
\be y_{k+\nu}- y^d_{k+\nu}= \cF_k+ \cG_k u_k - y^d_{k+\nu}. \label{erryd} \ee
Now define 
\be w_k := \cG_k u_k +\hat\cF_k. \label{wkdef} \ee
Substituting for $\cG_k u_k$ from eq. \eqref{wkdef} into eq. \eqref{erryd}, we obtain:
\be e^y_{k+\nu} = \cF_k + w_k -\hat\cF_k -y^d_{k+\nu} = w_k- e^\cF_k -y^d_{k+\nu}. 
\label{erryd2} \ee
After substituting the control law eq. \eqref{claw1} into expression \eqref{erryd2} and collecting terms, 
we get eq. \eqref{errbds}.
\end{proof}
An immediate corollary of this result follows.
\begin{corollary}\label{cor1}
If the uncertainty modeled by $\cF_k$ in the ultra-local model \eqref{ultramodel} is constant, then the 
control law \eqref{claw1} along with the observer \eqref{Fest1st} or \eqref{Fest2nd}, lead to 
stable finite-time convergence of tracking error $e^y_{k}$ to zero.
\end{corollary}
This follows immediately from Remark \ref{remcons} about the disturbance (ultra-local model) 
observers given in Section \ref{sec: ulm}, and Proposition \ref{contlaw1}. Although the control law \eqref{claw1} 
is simple to implement, note that it does not include direct feedback of the output tracking error, $e^y_k$. 
This may lead to lack of robustness to errors in measuring the output signal $y_k$. The following subsection 
gives another control law that does not have this drawback and is robust to output measurement errors.

\subsection{Second Output Trajectory Tracking Control Law}\label{ssec: trajcon2} 
\vspace*{-1mm}
The second tracking control law follows from:
\be y_{k+\nu}= \hat\cF_k +\cG_k u_k\, \mbox{ if } \, e^\cF_k=0. \label{0ecF} \ee
In this case, the output tracking error satisfies: 
\be e^y_{k+\nu}=  \hat\cF_k +\cG_k u_k - y^d_{k+\nu}. \label{ey0ecF} \ee
The statement below gives a tracking control law that includes feedback of the output tracking error.
\begin{theorem}\label{contlaw2}
Consider an unknown input-output system described by the ultra-local model \eqref{ultramodel}. Let 
$e^y_k$ be as defined by eq. \eqref{trackerrs}, and let $s \in ]1,2[$ and $\mu\in\bR^+$ be constants. 
Let the control law for the system be given by:
\begin{align} &\cG_k u_k= y^d_{k+\nu}- \hat\cF_k +\cC (e^y_{k+\nu-1}) e^y_{k+\nu-1}, \nn \\
&\mbox{where } \cC (e^y_j)= \frac{\big((e^y_j)\T e^y_j\big)^{1-1/s} -\mu}{\big((e^y_j)\T e^y_j
\big)^{1-1/s} +\mu}, \label{claw2} \end{align}
and the $y^d_j= y^d(t_j)$ describe a desired output trajectory for the time sequence $\{t_i\}\subset
\bR^+_0$ as in eq. \eqref{trackerrs}. Then the system \eqref{ultramodel} with the unknown dynamics 
$\cF^k$ along with the control law \eqref{claw2}, satisfies the error dynamics
\be e^y_{k+\nu}+ e^\cF_k =\cC (e^y_{k+\nu-1}) e^y_{k+\nu-1}. \label{errbds2} \ee
In particular, if the disturbance estimate $\hat\cF^k$ is obtained from the observer law \eqref{Fest1st} of 
Theorem \ref{prop1st} or \eqref{Fest2nd} of Theorem \ref{prop2nd}, then $e^y_{k+\nu}$ converges  
in a stable manner to a bounded neighborhood of $0\in\bR^n$ after finite time (i.e., for finite $k\in\bW$).
\end{theorem}
\begin{proof}
We begin the proof of this result by noting that if $e^\cF_k=0$, then $e^y_{k+\nu}$ satisfies eq. \eqref{ey0ecF}. 
In this situation, the estimation of $\cF_k$ is perfect, and finite-time stable convergence of the output tracking error 
to zero is guaranteed if:
\be  e^y_{k+\nu}= \cC (e^y_{k+\nu-1}) e^y_{k+\nu-1}, \label{errperf} \ee
with $\cC (e^y_j)$ defined as in eq. \eqref{claw2}. Note that the design of $\cC (e^y_j)$ is similar to the 
design of $\cD (e^\cF_k)$ in the observer law \eqref{Fest1st}. Defining the following Lyapunov function 
for the output tracking error:
\be V^y_k := (e^y_k)\T e^y_k, \label{Vydef} \ee
we can easily show that the error dynamics \eqref{errperf} leads to finite-time stable convergence 
of the output tracking error $e^y_k$ to zero; this would parallel the stability analysis in Theorem \ref{prop1st} 
for the model (disturbance) estimation error $e^\cF_k$. Based on Corollary \ref{cor1}, we know that perfect 
estimation of $\cF_k$ happens when $\cF_k$ is constant, for example; this is one case where $e^\cF_k=0$ 
for $k>N$ and a finite $N\in\bW$. In all such situations, eq. \eqref{errperf} ensures convergence 
of the output tracking error $e^y_k$ to zero in (an additional) finite amount of time. 

More generally, if $e^\cF_k$ reaches a bounded neighborhood of the zero vector in $\bR^n$, as would be 
the case if the observer laws given by either \eqref{Fest1st} of Theorem \ref{prop1st} or \eqref{Fest2nd} 
of Theorem \ref{prop2nd} are used, then the control law \eqref{claw2} leads to the following feedback 
dynamics:
\be y_{k+\nu}= y^d_{k+\nu}- e^\cF_k +\cC (e^y_{k+\nu-1}) e^y_{k+\nu-1}, \label{fbdyn} \ee
when substituted into the ultra-local model \eqref{ultramodel}. Re-arranging terms in eq. \eqref{fbdyn}, we 
get the error dynamics \eqref{errbds2}. Note that the error dynamics \eqref{errbds2} is a perturbation of 
the finite-time stable tracking error dynamics given by \eqref{errperf}, where the perturbation is due to the 
bounded error $e^\cF_k$ in estimating the unknown dynamics. 
\end{proof}

In the following subsection, we provide results on the robustness of this trajectory tracking control law 
in the presence of bounded estimation error $e^\cF_k$ in estimating the unknown dynamics $\cF_k$ and 
bounded measurement error (noise) $\eta_k$ in the measured outputs $y_k$.

\subsection{Robustness of Second Output Tracking Control Scheme}
\vspace*{-1mm}
Convergence of the output tracking error $e^y_k$ to zero for the control laws \eqref{claw1} in Proposition 
\ref{contlaw1} and \eqref{claw2} in Theorem \ref{contlaw2}, is contingent upon $e^\cF_k$ converging to zero 
in finite time, which happens in the special case that $\cF_k$ is constant according to Corollary \ref{cor1}. 
However, the usefulness of this control law \eqref{claw2} is its robustness to errors in the estimation of $\cF_k$, 
as given by the following result.   

\begin{proposition}\label{claw2prop}
Consider the feedback system consisting of the unknown system given by the ultra-local model \eqref{ultramodel}, 
the observer law \eqref{Fest1st} or the observer law \eqref{Fest2nd}, and the control law \eqref{claw2}. Let the 
estimation error $e^\cF_k$ in the unknown $\cF_k$ be bounded according to:
\be \|e^\cF_k\| \le B \mbox{ for } k>N, \label{errsbnd} \ee
where 
$B\in\bR^+$ and $N\in\bW$ are known. Then the tracking error $e^y_k$ converges to the 
neighborhood given by:
\be \cN^y :=\{ e^y_k\in\bR^n\, :\, \sigma (e^y_k) \|e^y_k\|\le B \}, \label{eynbhd} \ee
for $k>N'>N$, $N'\in\bW$, where 
\be \sigma (e_k^y) := \frac{\beta (e^y_k)}{1-\sqrt{1-\beta (e^y_k)}},\ 
\beta (e^y_k):=1-\big(\cC(e^y_k)\big)^2.  \label{sigdef} \ee
\end{proposition}
\begin{proof}
The proof of this result is similar to the proof of Theorem \ref{thecFbnd}. The feedback tracking error 
$e^y_k$ satisfies eq. \eqref{errbds2}. The first difference of the Lyapunov function \eqref{Vydef} is 
evaluated as follows:
\begin{align*}
& V^y_{k+1}- V^y_k = (e_{k+1}^y +e^y_k)\T  (e_{k+1}^y -e^y_k) \\
&=\big( (\cC (e^y_k))^2 -1\big) (e^y_k)\T e^y_k -2\cC(e^y_k) (e^\cF_k)\T e^y_k
+ (e^\cF_k)\T e^\cF_k. 
\end{align*}  
With the bound on $e^\cF_k$ given by \eqref{errsbnd} for $k>N$, we get an upper bound on the first 
difference of this Lyapunov function as follows:
\be V^y_{k+1}- V^y_k \le
-\beta (e^y_k) \|e^y_k\|^2 +2 |\cC(e^y_k)| B \|e^y_k\| +B^2, \label{delVybds} \ee
using the expression for $V^y_k$ in eq. \eqref{Vydef} and eq. \eqref{sigdef} for $\beta (e^y_k)$. 
The remainder of this proof follows the same last few steps as the proof of Theorem \ref{thecFbnd}, with 
the appropriate substitutions, i.e., $e^\cF_k$ replaced by $e^y_k$, $\cD(e^\cF_k)$ by $\cC(e^y_k)$, 
$B^\cF$ by $B$, and $\zeta(e^\cF_k)$ by $\beta(e^y_k)$. This leads to the conclusion that for $k>N'$
for some whole number $N'>N$, the tracking error $e^y_k$ converges to the neighborhood of the 
origin in $\bR^n$ given by eqs. \eqref{eynbhd}-\eqref{sigdef}. 
\end{proof}

In the presence of output measurement errors, an output observer can be implemented that filters out 
measurement noise and gives output estimates. The above result can then be extended to show robustness 
to both output estimation errors and model (disturbance) estimation errors in this situation, as the result 
below shows.

\begin{corollary}\label{robcor}
Consider the feedback system consisting of the unknown system given by \eqref{ultramodel}, with output 
measurements corrupted by bounded noise $\eta_k\in\bR^n$ as given by \eqref{measnoise}. Let a stable output 
observer provide output estimates $\hat y_k$ with bounded estimation errors $e^o_k:=\hat y_k- y_k$. 
Consider the feedback tracking control law eq. \eqref{claw2} with the ``observed" tracking error now defined 
as $\hat e^y_k:=\hat y_k- y^d_k$, used in conjunction with either of the ultra-local model 
observers given by eqs. \eqref{Fest1st} or \eqref{Fest2nd}. Then the resulting feedback system is (Lyapunov) 
stable and robust to errors in the ultra-local model estimate $e^\cF_k$ and the bounded observer error 
$e^o_k$. Further, if $e^\cF_k$ is bounded according to \eqref{errsbnd}, then the observed tracking error 
$\hat e^y_k$ converges to a neighborhood of the form $\hat\cN^y$ 
where 
\be \hat\cN^y :=\{ \hat e^y_k\in\bR^n\, :\, \sigma (\hat e^y_k) \|\hat e^y_k\|\le B \}, \label{heynbhd} \ee
for $k>N'>N$, $N'\in\bW$, where $\sigma(\cdot)$ is as defined by \eqref{sigdef} and $N$ is according to 
\eqref{errsbnd}. 
\end{corollary}
\begin{proof}
The proof of this result is identical to that of Proposition\ref{claw2prop}, with the substitution of the  
feedback control law \eqref{claw2} in terms of $\hat e^y_k$:
\be \cG_k u_k= y^d_{k+\nu}- \hat\cF_k +\cC (\hat e^y_{k+\nu-1}) \hat e^y_{k+\nu-1}, \label{claw2m} \ee
into the ultra-local model (ULM) given by eq. \eqref{ultramodel}.
\end{proof}

In the next section, we numerically apply this control scheme with the first-order ULM observer in Section 
\ref{ssec: estF1st}. 

\vspace*{-1mm}
\section{Numerical Simulation Results}\label{sec: numres}
\vspace*{-1mm}
Here we provide numerical simulation results of this model-free tracking control framework applied to an
inverted pendulum on a cart with nonlinear friction terms affecting the motion of both the degrees of freedom. 
The dynamics model of this system, unknown to the controller, is described in Section 
\ref{ssec: invpen}. The numerical results of the control scheme are given in Section \ref{ssec: simres}.
\subsection{Inverted pendulum on cart system}\label{ssec: invpen}
\begin{figure}[hbt]
\begin{center}
\hspace*{-2mm}\includegraphics[height=2.3in]{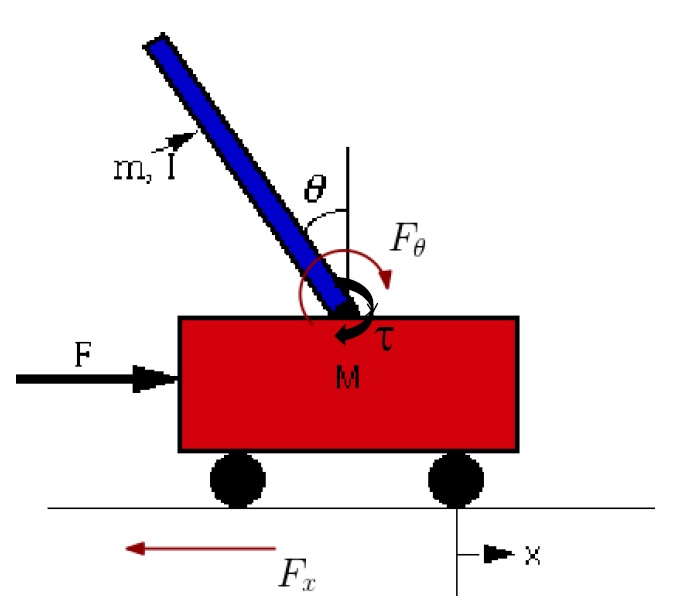}
\caption{Inverted pendulum system to which our nonlinear model-free control framework 
is applied.}  
\label{invpenfig}                                 
\end{center}        
\vspace*{-2mm}                         
\end{figure}
The inverted pendulum on cart is a two degree-of-freedom mechanical system, with the cart position $x$ 
considered positive to the right of a fixed origin and the angular displacement $\theta$ considered 
positive counter-clockwise from upward vertical, as shown in Fig. \ref{invpenfig}. The inputs to the 
system are a horizontal force on the cart denoted $F$ and a torque applied to the pendulum motor where it 
is attached to the cart denoted $\tau$. The outputs are the cart position $x$ and the angular 
displacement of the pendulum $\theta$. Therefore, this is a two-input and two-ouput system, unlike the usual 
single-input inverted pendulum on a cart example considered with only the horizontal cart force as an input. 
The mass and rotational inertia of the pendulum are $m$ and $I$ respectively, its length is $2l$, and the mass 
of the cart is $M$. A dynamics model of the system, which is unknown for the purpose of control design, 
is used to generate the desired output trajectory to be tracked. Then the model-free control scheme is 
used to track this desired trajectory.  

For simulation purposes, the inverted pendulum on a cart system is subjected to a nonlinear friction 
force acting on the cart's motion, and a nonlinear friction-induced torque acting on the pendulum. The 
friction force acting on the cart is denoted $F_x$ and the friction torque acting on the pendulum is 
denoted $F_\theta$, and they are given by:
\be F_x= c_x\tanh\dot x,\; F_\theta= c_\theta\tanh\dot\theta. \label{fricts} \ee
Note that the hyperbolic tangent function ensures that these frictional effects get saturated at high 
speeds ($\dot  x$ and $\dot\theta$). Therefore, the dynamics model of this system, which is unknown 
for the purpose of control design, is given by:
\begin{align}
\begin{split}
&\cM(q)\ddot q+ \cD(q,\dot q)= u,\; q= \bbm x\\ \theta \ebm, \\ 
&\cM(q)= \bbm M+m & -ml\cos\theta \\ -ml\cos\theta & I+ ml^2\ebm, \\ &\cD(q,\dot q)= 
\bbm ml\dot\theta^2 \sin\theta+ c_x\tanh\dot x \\ c_\theta\tanh\dot\theta - mgl\sin\theta \ebm.
\end{split} \label{invpenmodel}
\end{align}
The input and output vectors are:
\be u= \bbm F\\ \tau \ebm,\;\ y= q=\bbm x \\ \theta \ebm. \label{inpout-invpen} \ee
For the purpose of the numerical simulation, the parameter values selected for this system are:
\begin{align}
&M=1.5\, \mrm{ kg}, \; m=0.5\, \mrm{ kg},\; l=1.4\, \mrm{ m},\; I=0.84\, \mrm{ kg\, m}^2,\nn \\
&g= 9.8\, \mrm{ m/s}^2,\; c_x=0.028\, \mrm{ N}, \; c_\theta=0.0032\, \mrm{ N\, m}. 
\label{sysparams} 
\end{align}
The desired trajectory was generated by applying the following model-based control inputs (force and torque) to 
the cart and pendulum:
\begin{align} 
\begin{split}
 &F = ml\dot\theta^2\sin\theta - 2 (M+m\sin^2\theta) g\sin\theta\\ 
& - (M+m)g\sin\theta, \;\ 
\tau = -mgl\sin\theta.
\end{split}\label{OLcont} 
\end{align}
This generates an output trajectory $y^d(t)=[x^d(t)\; \theta^d (t)]\T$ that is oscillatory, as depicted in 
Fig. \ref{gentraj} in Section \ref{ssec: simres}. Note that this dynamics model and model-based control 
used here is purely for the purpose of trajectory generation and to demonstrate the working of the 
model-free control framework outlined in this paper. 

\vspace*{-1mm}
\subsection{Discretization of continuous dynamics model}\label{ssec: simdisc}
\vspace*{-1mm}
The dynamics model and control law for the inverted pendulum on cart system given in Section 
\ref{ssec: invpen}, are discretized here using forward difference
 schemes for generalized velocities and 
accelerations of the two degrees of freedom. Denoting outputs and inputs in discrete time by $y_k:= 
q_k=q(t_k)$ and $u_k:= u(t_k)$ as before and the time step size by $\Delt:=t_{k+1}-t_k$, we get the 
following discretization of the contunuous dynamics \eqref{invpenmodel},
\begin{align} 
\frac{y_{k+2}-2 y_{k+1}+ y_k}{\Delt^2}=& \big(\cM(y_k)\big)^{-1}\Big(u_k - \nn \\ 
&\cD\big(y_k,\frac{y_{k+1}-y_k}{\Delt}\big)\Big), \label{discinvpen} 
\end{align}
where $\cM(\cdot)$ and $\cD(\cdot,\cdot)$ are as defined in eq. \eqref{invpenmodel}. This leads to 
the following second order discrete-time system:
\begin{align} 
& y_{k+2}= \sF_k + \sG_k u_k, \mbox{ where } \sG_k= \Delt^2 \big(\cM(y_k)\big)^{-1} \label{dinvpen} \\
& \mbox{and } \sF_k= 2 y_{k+1}- y_k - \Delt^2\big(\cM(y_k)\big)^{-1}\cD\big(y_k,\frac{y_{k+1}-y_k}
{\Delt}\big), \nn
\end{align}
where $\sF_k$ and $\sG_k$ have the meanings as defined by eq. \eqref{truedyn}. 
In the numerical simulation results shown in the next subsection, this discrete time system is used for 
generating the desired output trajectory starting from a given initial state vector and with the control laws 
given by eqs. \eqref{OLcont} sampled at time instants $t_k$. It is then used to simulate the performance 
of the data-driven control approach outlined in Sections \ref{sec: ulm} and \ref{sec: mfccon} with the 
discrete dynamics \eqref{dinvpen} unknown to the control law.

\vspace*{-1mm}
\subsection{Simulation results for control scheme}\label{ssec: simres}
\vspace*{-1mm}
Here we present numerical simulation results for the model-free tracking control scheme applied 
to the system described by eqs. \eqref{invpenmodel}-\eqref{sysparams}. A trajectory is generated 
for this system using the control scheme \eqref{OLcont} sampled in discrete time, with the initial states: 
\be \bbm q^d(0)\\ \dot q^d(0) \ebm= \bbm x^d(0)\\ \theta^d (0)\\ \dot x^d(0)\\ \dot\theta^d(0) \ebm= 
\bbm 0.45\, \mbox{m}\\ -0.14\, \mbox{rad} \\ -0.3\, \mbox{m/s} \\ 0.05\, \mbox{rad/s} \ebm. \label{initstates} \ee
The generated trajectory $y^d(t)=\theta^d (t)$ for a time interval of $T=70$ seconds is depicted in 
Fig. \ref{gentraj}. 
\begin{figure}[htb]
\begin{center}
\includegraphics[width=0.91\columnwidth]{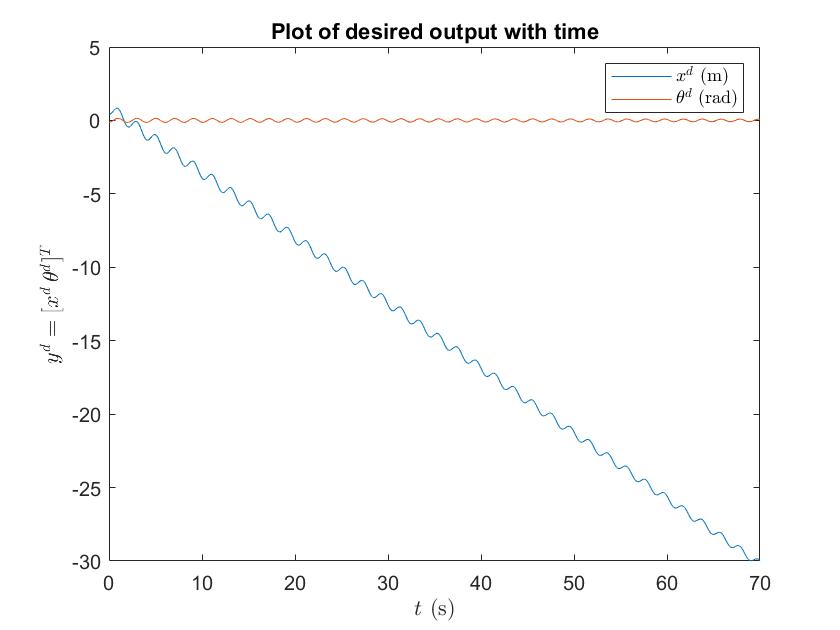}
\vspace*{-1mm}
\caption{Desired trajectory generated for $T=70$ seconds for inverted pendulum on cart system.}  
\label{gentraj}                                 
\end{center}        
\vspace*{-2mm}                         
\end{figure}

The control scheme given by Theorem \ref{contlaw2} is applied to this system to track this  
desired trajectory. For this simulation, the initial estimated states are: 
\be \bbm \hat q(0)\\ \dot{\hat q}(0) \ebm= \bbm \hat x(0)\\ \hat \theta(0)\\ \dot{\hat x}(0)\\ \dot{\hat\theta}(0) 
\ebm= \bbm 0\, \mbox{m}\\ 0.102\, \mbox{rad} \\ 0\, \mbox{m/s} \\ 0\, \mbox{rad/s} \ebm. \label{initcs} \ee
Output measurements are assumed at a constant rate of 100 Hz, i.e., sampling period $\Delt= 0.01$ s. 
In the simulation, the measurements are generated by numerically propagating the true discrete-time dynamics 
of the inverted pendulum on cart system given by eqs. \eqref{dinvpen}, and adding noise to the true outputs 
$y_k=q_k$. The additive noise is generated as high frequency and low amplitude sinusoidal signals, where the 
frequencies are also sinusoidally time-varying.
A finite-time stable output observer given by:
\begin{align} 
\begin{split}
&\hat y_{k+1}= y_{k+1}+ \mathcal{B}(e^o_k) e^o_k, \, \mbox{ where } e^o_k=\hat y_k -y_k, \\ 
&\mbox{and } \cB (e^o_k)= \frac{\big((e^o_k)\T Le^o_k\big)^{1-1/p} - \beta}{\big((e^o_k)\T Le^o_k
\big)^{1-1/p} + \beta}. 
\end{split} \label{dFTS-obs} 
\end{align}
with the observer gains: 
\be L =2.1,\;\ \beta=2,\; \mbox{ and }\; p=\frac75, \label{obsv-gains} \ee
is usedto filter out noise from the measured outputs. 
The first order ultra-local model observer given by Theorem \ref{prop1st} is used, with observer gains:
\be \lambda=1.5,\; \mbox{ and }\; r=\frac97. \label{ulmobs-gains} \ee
This observer is initialized with the zero vector, i.e., $\hat\cF_0=0$.
The control law \eqref{claw2} is then used to compute the control inputs $u_k$. The control 
gains used in this simulation are:
\begin{align} 
s=\frac{11}{9}, \; \mu=0.35, \mbox{ and } \cG_k= \Delt \matl{cc} 0.559 & 0.196 \\ 0.196 & 0.657 \matr,
\label{contr-gains} 
\end{align}
where $\cG_k$ is selected to be symmetric and positive definite based on the expected form for a mechanical 
system. 

\begin{figure}[htb]
\begin{center}
\includegraphics[width=0.91\columnwidth]{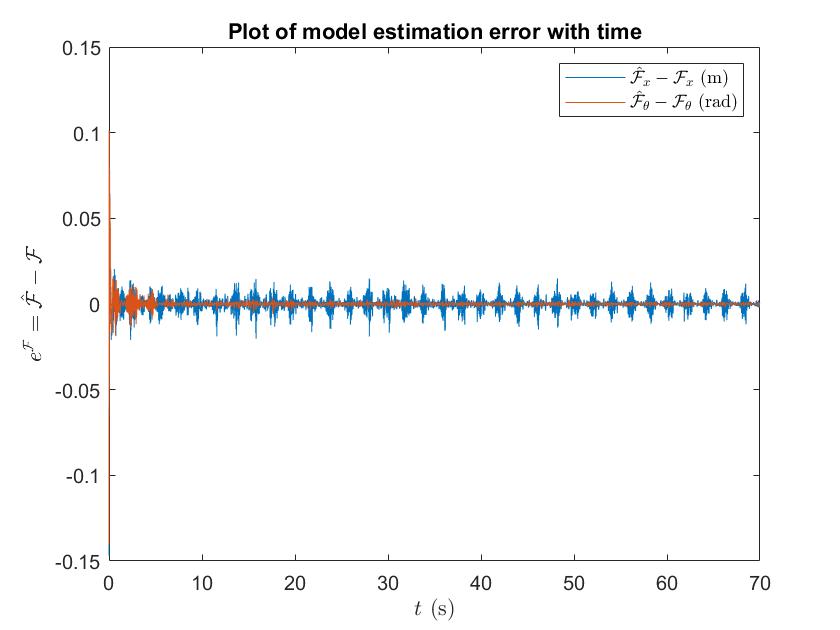}
\vspace*{-5mm}
\caption{Estimation error in ultra-local model estimation (bottom) for inverted 
pendulum on cart system with model-free control.}  
\label{estimerrs}                                 
\end{center}        
\end{figure}
Simulation results for the estimation error in estimating the 
unknown $\cF_k$ according to the observer given by Theorem \ref{prop1st} in Section \ref{ssec: estF1st}, is 
depicted in Fig. \ref{estimerrs}. Simulation results for the tracking control performance and control input are 
shown in Fig. \ref{tracontinp}. 
\begin{figure}[hbt]
\begin{center}
\includegraphics[width=0.91\columnwidth]{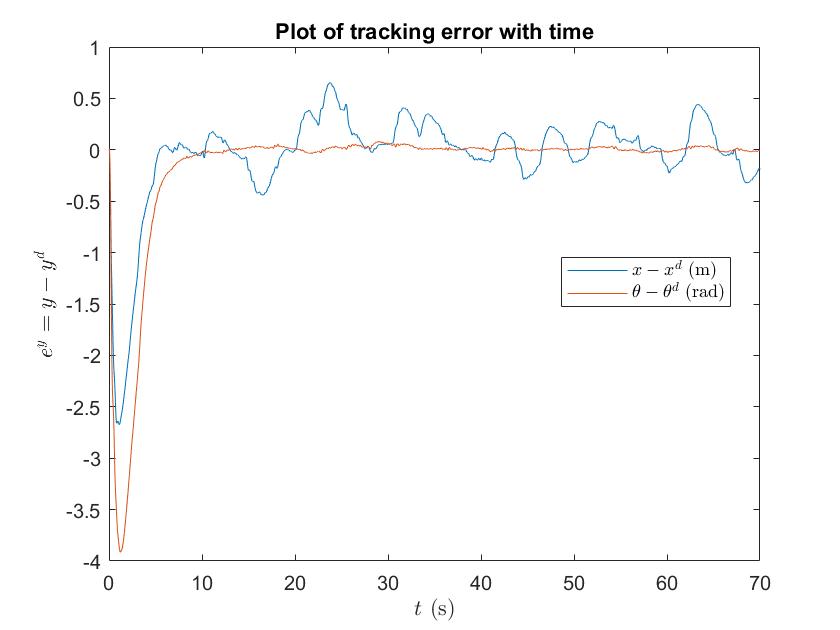}
\vspace*{1mm}
\includegraphics[width=0.91\columnwidth]{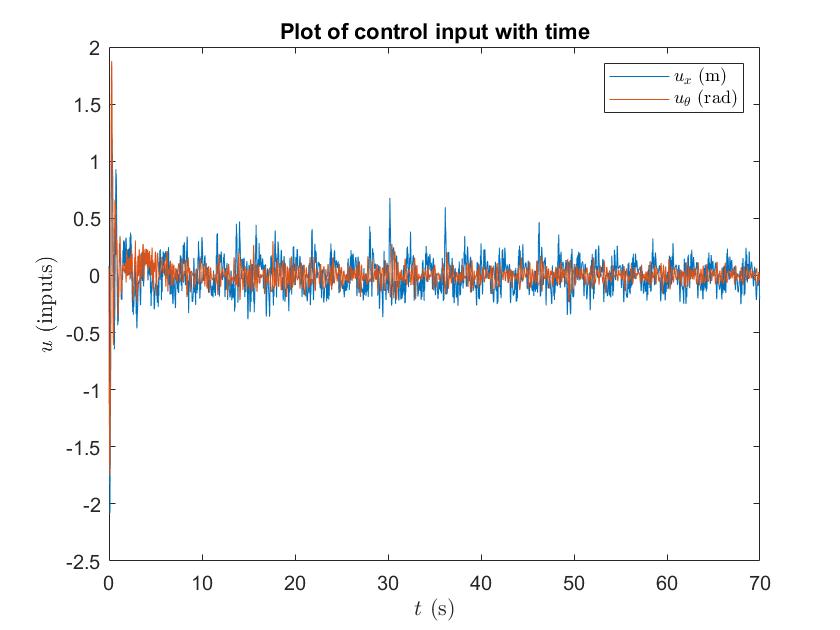}
\vspace*{-5mm}
\caption{Output trajectory tracking error (top) and control input (bottom) for inverted pendulum on cart 
system with model-free control.}  
\label{tracontinp}                                 
\end{center}        
\end{figure}
The plot on the top shows the output trajectory tracking error over the simulated duration. Note that the 
tracking error settles down to within an error bound less than about $0.5$ m in cart position and $0.05$ rad 
in pendulum angle in steady state, after an initial brief period of transients. The time plot of the control inputs 
is shown in the bottom plot.
This control input profiles show some high frequency oscillations in tracking the desired trajectory, that 
seem to correlate with the oscillations seen in 
the ultra-local model observer error in Fig. \ref{estimerrs}. Future work will deal 
with reducing these transients by 
using integral term(s) in the observer designs to produce smoother estimates of the output and ultra-local 
model. A reference governor may also be used to modify the reference (desired) output trajectory based 
on current estimates of outputs as in~\cite{ergov16}.  

\begin{remark}
Although the schemes given here assume that the output space is a vector space, the angle output 
for this inverted pendulum on cart example is on the circle $\bS^1$, which is not a vector space. 
Therefore, the observer and control laws outlined in the earlier sections may lead to unwinding, 
even though that does not happen for the numerical simulation reported here. The model-free observer 
and controller design framework outlined here will be extended to systems evolving on non-Euclidean 
output (or state) spaces in the future, to address this issue.
\end{remark}

\vspace*{-1mm}
\section{Conclusion}\label{sec: conc}
\vspace*{-2mm}
This paper presents a formulation of a data-driven (model-free) control approach that guarantees nonlinear 
stability and robustness for output tracking control with feedback of output measurements that may contain 
additive noise. The formulation presented here is developed in discrete time, and uses the concept of an 
ultra-local model used to model unknown input-output behavior, similar to the linear model-free 
control approach formulated in the last decade. 
This formulation begins with a finite-time stabilization scheme in discrete time that leads to a 
H\"{o}lder-continuous feedback system. This finite-time stabilization scheme is then used to develop 
nonlinearly stable and robust observers to estimate in real time the ultra-local model that models the unknown 
input-output dynamics, from past input-output history. The estimates of the unknown dynamics are then 
used for compensation of these unknowns (considered as a disturbance input) in a nonlinear output feedback 
tracking control law that is designed to track a desired output trajectory that is smooth, in a nonlinearly 
stable and robust manner. Nonlinear stability analysis shows the stability of the feedback compensator 
combining the nonlinear observer and nonlinear control law when the change in the discrete-time system 
dynamics modeled by the ultra-local model has a bounded finite difference. A numerical simulation 
experiment is carried out on an inverted pendulum on a cart system with nonlinear friction, for which the 
inputs are the horizontal force applied to the cart and a torque applied to the pendulum, and the outputs 
are the cart horizontal displacement and angular displacement of the pendulum from the upward vertical. 
Noisy measurements of the outputs are available with bounded amplitude of noise. The model of the 
dynamics of this system is unknown to the nonlinear observers and controller designed using our nonlinear
model-free control framework. This numerical experiment shows convergence of output estimation errors 
and output tracking errors to small absolute values. Future work will explore extensions of this framework 
to systems evolving on Lie groups and their principal bundles, and also development of stable higher-order 
observers for the ultra-local model for increased robustness.
\section{Acknowledgements}
The author acknowledges support from the Systems and Control (SysCon) department at Indian Institute 
of Technology, Bombay, India, (IIT-B), where he was hosted in the summer of 2019 when commencing 
this work. In particular, helpful discussions with Debasish Chatterjee and Sukumar Srikant at SysCon (IIT-B), 
are gratefully acknowledged. 

\bibliographystyle{model-1-names}
\bibliography{alias,references}


\end{document}